\documentclass{amsart}
\usepackage{amsmath}
\usepackage{amssymb}
\usepackage{amsthm}
\usepackage[all]{xy}

\usepackage{changepage}

\newtheorem{theorem}{Theorem}
\newtheorem{lemma}[theorem]{Lemma}

\newtheorem{conjecture}[theorem]{Conjecture}
\newtheorem{corollary}[theorem]{Corollary}
\newtheorem*{claim}{Claim}
\newtheorem{theoremm}{Theorem}

\theoremstyle{definition}
\newtheorem{definition}[theorem]{Definition}
\newtheorem{example}[theorem]{Example}
\newtheorem{question}[theorem]{Question}

\newcommand{\Z}{\mathbb{Z}}
\newcommand{\R}{\mathbb{R}}

\theoremstyle{remark}
\newtheorem{remark}[theorem]{Remark}

\title{The word problem of $\Z^n$ is a multiple context-free language}
\author{Meng-Che ``Turbo" Ho}
\thanks{The material is based upon work supported by the National Science Foundation under Grant No. DMS-1440140 while the author was in residence at the Mathematical Sciences Research Institute in Berkeley, California, during the Fall 2016 semester.}
\date{\today}

\begin{document}

\begin{abstract}
The \emph{word problem} of a group $G = \langle \Sigma \rangle$ can be defined as the set of formal words in $\Sigma^*$ that represent the identity in $G$. When viewed as formal languages, this gives a strong connection between classes of groups and classes of formal languages. For example, Anisimov showed that a group is finite if and only if its word problem is a regular language, and Muller and Schupp showed that a group is virtually-free if and only if its word problem is a context-free language. Above this, not much was known, until Salvati showed recently that the word problem of $\mathbb{Z}^2$ is a multiple context-free language, giving first such example. We generalize Salvati's result to show that the word problem of $\mathbb{Z}^n$ is a multiple context-free language for any $n$. 
\end{abstract}

\maketitle

\section{Introduction}

Fixing an alphabet set $\Sigma$, the set of finite strings in $\Sigma$ is denoted by $\Sigma^* = \{ a_1a_2\cdots a_n \mid n \in \mathbb{N}, a_i \in \Sigma\}$. Any subset of $\Sigma^*$ is called a \emph{formal language}. As an attempt to analyze natural languages, Chomsky \cite{Ch56} introduced the Chomsky-Sch\"utzenberger hierarchy, which includes \emph{regular}, \emph{context-free}, \emph{context-sensitive}, and \emph{recursively enumerable} languages, each one a larger class than the former. In this paper, we will be focusing on \emph{multiple context-free languages}, another generalization of context-free languages introduced by Seki et al \cite{Se91}.

There are many ways to associate formal languages to a group. Indeed, for a group $G$ generated by a finite set $\Sigma = \Sigma^{-1}$, we define its \emph{word problem} $W(G)$ to be the subset of $\Sigma^*$ consisting of all strings that represent the identity element of $G$. Note that the language (and even the set of alphabets) depends on the choice of generating set. However, we will see that this dependence is negligible for the languages we are interested in here.

There had been various work associating classes of languages and classes of groups via word problems. An{\= \i}s{\=\i}mov \cite{An71} showed that $G$ is finite if and only if $W(G)$ is regular. Muller and Schupp showed that $G$ is virtually free if and only if $W(G)$ is context-free. Indeed, they showed the following stronger theorem:

\begin{theorem}[Muller, Schupp \cite{Mu83}]
For a finitely generated group $G$, the following are equivalent:
\begin{enumerate}
\item $G$ is virtually free.
\item $W(G)$ is context-free.
\item There exists a constant $K$ such that every closed path in the Cayley graph of $G$ can be triangulated by chords with length $\le K$.
\end{enumerate}
\end{theorem}

Recently, Salvati \cite{Sa15} showed that the word problem of $\Z^2$ is a 2-multiple context-free language, giving first such example. In the proof, he used some geometric methods that seems to be specific to the two-dimensional case. Naderhof \cite{Ne16} gave a shorter proof of Salvati's result, avoiding some of the pieces that are two-dimensional specific. In this paper, by using a topological lemma from Burago \cite{Bu92}, we also avoid the geometric pieces and are able to generalize Silvati's theorem and show that the word problem of $\Z^n$ is multiple context-free:

\begin{theorem}
The word problem $W(\Z^n)$ is an $(8\left[\frac{n+1}{2}\right]-2)$-multiple context-free language for every $n$.
\end{theorem}

However, this yields that $W(\Z^2)$ is a $6$-multiple context-free language, which is slightly weaker than Salvati's result which says that $W(\Z^2)$ is a $2$-multiple context-free language. Thus, one natural question will be improving our result to get the following conjecture proposed by Nederhof:

\begin{conjecture}[Nederhof \cite{Ne17}]
For any $n$, the language $W(\Z^n)$ is $n$-multiple context-free.
\end{conjecture}

In the same paper, Nederhof also showed that $W(\Z^n)$ is not a $K$-multiple context-free language for any $K<n$, so this conjecture is indeed sharp. 

To end the introduction, we propose the following open problem that may be interesting to group theorists. For more recent progress on this topic, the reader may also want to see \cite{Gi} and \cite{Kr}.

\begin{question}
For which finitely-generated groups $G$ is $W(G)$ multiple context-free?
\end{question}

This paper is organized as follows: We give the necessary definition and background in section \ref{dab}. We then set up our notation and introduce the main lemma in section \ref{naml}. Finally we show the main Theorem in section \ref{mt}.

\section{Definition and background}\label{dab}

We follow \cite{Se91} and \cite{Ne17} to give the definition of multiple context-free grammars:

\begin{definition}
A \emph{$K$-multiple context-free grammar} is a tuple $(\Sigma, N, S, P)$, where $\Sigma$ is the set of \emph{terminals}, $N$ is the set of \emph{non-terminals}, $S\in N$ is the \emph{start symbol}, and $P$ is the set of \emph{production rules}, such that:
\begin{itemize}
\item Each non-terminal $I(x_1, x_2,\ldots,x_m)$ has an associated natural number $m \le K$, its \emph{arity}.
\item The arity of $S$ is 1.
\item Each of the production rule is of the following form for some $p\ge 0$:
$$ I_0(x_1,x_2,\ldots,x_{m_0}) \leftarrow I_1(y_1^1,y_2^1,\ldots,y_{m_1}^1), I_2(y_1^2,\ldots,y_{m_2}^2), \cdots, I_p(y_1^p,\ldots,y_{m_p}^p) $$
where $I_i$ has arity $m_i$, each $x_i$ is a word in the variables $y_j^l$ and the terminals $\Sigma$, and each of $y_j^l$ shows up at most once in all of $x_1,x_2,\ldots, x_n$.
\end{itemize}

An \emph{instance} of a production rule or non-terminal is obtained by replacing the occurance(s) of each variable by a string in $\Sigma^*$. A sequence of rule instances is \emph{valid} if every non-terminal instances on the right hand side of some rule instances in the sequence already appears on the left hand side of some rule instances earlier in the sequence. We say a sequence \emph{ends} in a non-terminal instance if it is the left hand side of the last rule instance in the sequence.

We say a grammar \emph{admits} a string $s$ in $\Sigma^*$ if there is a valid sequence of rule instances, called a \emph{derivation}, ending in $S(s)$. We say a grammar \emph{produces} the language which is the collection of all strings in $\Sigma^*$ that it admits.
\end{definition}

\begin{definition}
A language is called \emph{$K$-multiple context-free} if it can be produced by some $K$-multiple context-free grammar. A language is called \emph{multiple context-free} if it is $K$-multiple context-free for some $K$.
\end{definition}

\begin{example}
Consider the 2-multiple context-free grammar with terminals $a, b, c, d$, a binary non-terminal $I = I(x,y)$, the start symbol $S = S(x)$, and production rules
$$ I(\epsilon, \epsilon) \leftarrow $$
$$ I(axb,cyd) \leftarrow I(x,y) $$
$$ S(xy) \leftarrow I(x,y).$$

Then the following is a derivation for $a^nb^nc^nd^n$:
$$ I(\epsilon, \epsilon) \leftarrow $$
$$ I(ab,cd) \leftarrow I(\epsilon, \epsilon) $$
$$ I(a^2b^2,c^2d^2) \leftarrow I(ab,cd) $$
$$\vdots$$
$$ I(a^nb^n,c^nd^n) \leftarrow I(a^{n-1}b^{n-1},c^{n-1}d^{n-1}) $$
$$ S(a^nb^nc^nd^n) \leftarrow I(a^nb^n,c^nd^n)$$

Indeed, this is the only kind of strings this grammar admits, so the language produced by this grammar is $\{ a^nb^nc^nd^n \mid n \ge 0\}$.

\end{example}

The class of multiple context-free languages is a cone \cite{Se91} in the sense of \cite{Gi}. A well-known consequence of being a cone is the following theorem from \cite{Gi}, which we only state in the case of multiple context-free languages:

\begin{theorem}\label{finite-index}
Let $W(G)$ be the word problem of $G$ with respect to a fixed generating set $\Sigma$. Suppose $W(G)$ is multiple context-free. Then
\begin{enumerate}
\item The word problem of $G$ is multiple context-free with respect to any generating set of $G$.
\item The word problem of any finitely-generated subgroup of $G$ is multiple context-free.
\item The word problem of any supergroup of $G$ of finite index is multiple context-free.
\end{enumerate}
\end{theorem}

\section{Notations and main lemma}\label{naml}

For any curve $\alpha$ in $\mathbb{R}^n$, we will write $\overline{\alpha}$ to be the displacement vector of $\alpha$, and $|\alpha|$ to be the length of $\alpha$. Note that when $\alpha$ is the curve on the Cayley graph of $\Z^n$ corresponding to a word $w$, $\overline{\alpha}$ is just the image of the word $\overline{w}$ in $\Z^n$ and $|\alpha|$ is just the length of the word. We also write $\alpha\beta$ to mean the concatenation of the two curves $\alpha$ and $\beta$.

We will need the following lemma by Burago to prove the main theorem. We include the proof here for completeness.

\begin{lemma}[Burago \cite{Bu92}]\label{main lemma}
Let $\alpha(t):[0,1]\to \mathbb{R}^n$ be a continuous curve. Then there exist $[t_1,s_1], [t_2,s_2], \ldots,[t_k,s_k]$ such that $0\le t_1\le s_1 \le t_2\le s_2\le  \cdots\le t_k\le s_k\le 1$ with $k= [\frac{n+1}{2}]$ and 
$$\sum\limits_{i=1}^k (\alpha(s_i)-\alpha(t_i)) = \frac{1}{2}(\alpha(1)-\alpha(0)).$$
\end{lemma}

\begin{proof}
We will construct an odd map $f:S^n\to\mathbb{R}^n$. For $\overline{x} = (x_1,x_2,\ldots,x_{n+1})\in S^n \subset \mathbb{R}^n$, we write $y_0 = 0$, $y_1 = x_1^2$, $y_2 = x_1^2+x_2^2$, \ldots, $y_{n+1} = x_1^2 + \cdots + x_{n+1}^2 = 1$. We define $$f(\overline{x}) = \sum\limits_{i=1}^{n+1}\operatorname{sign}(x_i)(\alpha(y_i) - \alpha(y_{i-1})).$$

One checks that $f$ is indeed an odd map, and since every odd map from $S^n \to \mathbb{R}^n$ has a zero, we can pick some $\overline{x}\in S^n$ such that $f(\overline{x}) = 0$. Then one of $\{ [y_{i-1},y_i]\mid x_i>0\}$ and $\{ [y_{i-1},y_i]\mid x_i<0\}$ satisfies the condition, by possibly attaching some empty intervals in the end.
\end{proof}

\begin{remark}
One sees that $k= [\frac{n+1}{2}]$ is sharp by considering the curve along the edges of an $n$-cube.
\end{remark}

\section{Proof of the main Theorem}\label{mt}

We restate and prove the main Theorem. To make the proof clearer, we shall suppress the embedding when talking about curves.

\setcounter{theoremm}{1}
\begin{theoremm}
The word problem $W(\Z^n)$ is an $(8\left[\frac{n+1}{2}\right]-2)$-multiple context-free language for every $n$.
\end{theoremm}

\begin{proof}
We will denote the standard generators of $\Z^n$ by $a_1, a_2, \ldots, a_n$, and their inverses by $A_1, A_2, \ldots, A_n$.

Let $k = \displaystyle\left[\frac{n+1}{2}\right]$ and $m = \displaystyle 8k-2$. We consider the $m$-multiple context-free grammar with terminals $\Sigma = \{ a_1, A_1, a_2, A_2, \ldots, A_n \}$, an $m$-ary non-terminals $I = I(x_1,x_2,\ldots,x_m)$, start symbol $S = S(x)$, and production rules
$$S(x_1x_2\cdots x_m) \leftarrow I(x_1,x_2,\ldots,x_m)$$
$$I(\epsilon,\epsilon,\ldots,\epsilon)\leftarrow$$
$$I(a_1,A_1,\epsilon,\epsilon,\ldots,\epsilon)\leftarrow$$
$$I(a_2,A_2,\epsilon,\epsilon,\ldots,\epsilon)\leftarrow$$
$$\vdots$$
$$I(a_n,A_n,\epsilon,\epsilon,\ldots,\epsilon)\leftarrow$$
$$I(z_1,z_2,\ldots,z_m)\leftarrow I(x_1,x_2,\ldots,x_m), I(x_{m+1},x_{m+2},\ldots,x_{2m})$$
for every $\sigma\in \operatorname{Sym}(2m)$ and $z_1z_2\cdots z_m = x_{\sigma(1)}x_{\sigma(2)}\cdots x_{\sigma(2m)}$.

By the very first production rule, it suffices to show that $x_1x_2\cdots x_m \in W(\Z^n)$ if and only if there is a valid sequence ending in $I(x_1,x_2,\ldots,x_m)$. The if direction follows immediately from the definition of the production rules and commutativity of $\Z^n$. To show the only if direction, we will induct on the length of $x_1x_2\cdots x_m \in W(\Z^n)$. The base cases where $|x_1\cdots x_m|\leq m$ are straight forward. Now suppose $|x_1\cdots x_m| > m$. We may assume that $x_1x_2\cdots x_{m/2} \neq 0$ and $x_{m/2+1}x_{m/2+2}\cdots x_{m} \neq 0$, otherwise we may apply induction hypothesis on $x_1x_2\cdots x_{m/2} = 0$ and $x_{m/2+1}x_{m/2+2}\cdots x_{m} = 0$.

Now for every $i$, we write $v_i$ to be the curve in $\R^n$ that represents $x_i$, which is contained in the Cayley graph of $\Z^n$. Applying Lemma \ref{main lemma} to the curve $v_1\cdots v_{m/2}$, we get $2k$ points on the curve, breaking the curve into $2k+1$ pieces $u_0,u_1,\dots,u_{2k}$, with

\begin{align}
\sum\limits_{i = 1}^{i = k} \overline{u_{2i-1}} = \frac12\sum\limits_{i = 0}^{i = 2k} \overline{u_{i}} = \sum\limits_{i = 0}^{i = k} \overline{u_{2i}}.
\end{align}

Applying Lemma \ref{main lemma} to the curve $v_{m/2+1}\cdots v_{m}$ in the same way, we get 
\begin{align}
\sum\limits_{i = 1}^{i = k} \overline{u'_{2i-1}} = \frac12\sum\limits_{i = 0}^{i = 2k} \overline{u'_{i}} = \sum\limits_{i = 0}^{i = k} \overline{u'_{2i}}.
\end{align}

\begin{claim}
We may choose the endpoints of all the $u_i$'s and $u'_i$'s to be either on lattice points, or \emph{mid-lattice points}, i.e.~the midpoints of two adjacent lattice points.
\end{claim}

\begin{proof}
Fix $u_i$'s and $u'_i$'s that has the fewest number of endpoints not being a lattice point or a mid-lattice point. Suppose the claim is false, then there is some endpoint $x$ that is not a lattice point or a mid-lattice point. Without loss, say this is an endpoint of $u_i$, and its first coordinate is not a multiple of $\frac12$. This forces all its other coordinates to be integral, as $u_1u_2\cdots u_m$ is a curve contained in the lattice $\Z^n$.

But each coordinate of the middle term of (1) is a multiple of $\frac12$, thus there must be another endpoint $y$ of some $u_j$ (possibly with $i=j$) whose first coordinate is also not a multiple of $\frac12$. Now depending on how $x$ and $y$ appear in (1), either we may replace $x$ by $x+\delta e_1$ and $y$ by $y + \delta e_1$ and still keep (1), or we may replace $x$ by $x+\delta e_1$ and $y$ by $y - \delta e_1$ and still keep (1). Say we are in the first case, then choose the smallest positive $\delta$ that makes $x+\delta e_1$ or $y+\delta e_1$ a multiple of $\frac12$, and replace $x$ and $y$ correspondingly. This gives another choice of $u_i$'s and $u'_i$'s which has at least one fewer endpoint not being a lattice point or a mid-lattice point, a contradiction.
\end{proof}

Now, the $2k$ points together with the endpoints of $v_i$'s, which are lattices points by definition, give a refinement $w_1,w_2,\ldots,w_{m/2+2k}$ and $w'_1,w'_2,\ldots,w'_{m/2+2k}$ of the $u_i$'s and $u'_i$'s. Notice that some of these may be empty curves if some of these points coincide. Now rewriting (1) and (2) in terms of $w_i$'s and $w'_i$'s, we get 
\begin{align*}
\sum\limits_{i\in S} \overline{w_{i}} = \frac12\sum\limits_{i = 1}^{i = m/2} \overline{x_{i}} = \sum\limits_{i \in \overline{S}} \overline{w_{i}}
\end{align*}
\begin{align*}
\sum\limits_{i\in T} \overline{w'_{i}} = \frac12\sum\limits_{i = m/2+1}^{i = m} \overline{x_{i}} = \sum\limits_{i \in \overline{T}} \overline{w'_{i}}
\end{align*}
for some $S,T\subset \{1,2,\ldots, m/2+2k\}$, and $\overline{S} = \{1,2,\ldots, m/2+2k\}\setminus S$, $\overline{T} = \{1,2,\ldots, m/2+2k\}\setminus T$. Notice that since $w_i$'s and $w'_i$'s are refinements of $u_i$'s and $u'_i$'s, so we have $|S|, |T|, |\overline{S}|, |\overline{T}| \geq k$. Without loss, we will assume that $|S| \geq |\overline{S}|$, and $|T| \leq |\overline{T}|$. Thus $|S|, |\overline{T}| \leq m/2+2k-k = 5k-1$, and $|\overline{S}|, |T| \leq (m/2+2k)/2 = 3k-1/2$, so $|\overline{S}|, |T| \leq 3k-1$. Hence $|S|+|T| \leq 8k-2 = m$, and $|\overline{S}|+|\overline{T}| \leq 8k-2 = m$.

Now observe that $\sum\limits_{i = 1}^{i = m} \overline{x_{i}} = 0$ by the assumption that $x_1x_2\cdots x_m \in W(\Z^n)$. Thus we have
\begin{align}
\sum\limits_{i\in S} \overline{w_{i}} + \sum\limits_{i\in T} \overline{w'_{i}} = \sum\limits_{i \in \overline{S}} \overline{w_{i}} + \sum\limits_{i \in \overline{T}} \overline{w'_{i}} = 0.
\end{align}

\begin{claim}
There is a partition $\{ w_i \mid 1 \leq i \leq m/2+2k\}$ of $v_1v_2\cdots v_{m/2}$ and $\{ w'_i \mid 1 \leq i \leq m/2+2k\}$ of $v_{m/2+1}\cdots v_{m}$ which is a refinement of $x_i$'s such that:
\begin{itemize}
\item[(a)] This partition satisfies (3).
\item[(b)] Not all curves in the first part of (3) are empty. Also, not all curves in the second part of (3) are empty.
\item[(c)] All endpoints are lattice points.
\end{itemize}
\end{claim}

\begin{proof}
The previous construction gives a partition satisfying (a) with every endpoint being either a lattice point or a mid-lattice point. It also satisfies (b) by (1) and (2) and the assumption that $x_{1}x_{2}\cdots x_{m/2}\neq 0$.Thus, we may choose a partition satisfying (a) and (b) whose endpoints are either a lattice point or a mid-lattice point, with fewest number of endpoints being a mid-lattice point. Supposing the claim is false, let $x$ be an endpoint that is a mid-lattice point. If both of the curves adjacent to it are in the same sum, we may simply replace $x$ to one of the closest lattice point, which reduces the number of endpoints being a mid-lattice point in our partition while still satisfying (a) and (b), a contradiction.

Now without loss, we may assume $x$ is adjacent to $w_1$ and $w_2$, with $1\in S$ and $2\notin S$, and assume $x$ is not integral in the first coordinate, with all other coordinates being integral. Note that in this case $x$ contributes positively to the first part of (3), and negatively to the second part of (3). As the sums equal to 0, there must be another endpoint $y$ which is also not integral in the first coordinate. By the argument in the previous paragraph, the two curves adjacent to $y$ cannot be in the same sum. Thus, similar to the case of $x$, $y$ must contribute to either the first or second part of (3) positively, and the other part negatively. 

We now replace $x$ with $x+\frac12 e_1$ and $y$ with $y\pm\frac12 e_1$, depending on whether $y$ contributes positively or negatively to the second part of (3). This new partition still satisfies (a). If this replacement breaks (b), say all curves in the first part of (3) are empty, then this replacement must have reduced the total length of curves in the first part of (3) by 1. So, we may instead replace $x$ with $x-\frac12 e_1$ and $y$ with $y\mp\frac12 e_1$, and increase the the total length of curves in the first part of (3) by 1, while reducing the total length of the second part by 1. However, as the total length of all curves are $>m$, this second replacement will satisfy (b).

Either case, this will replace the mid-lattice points $x$ and $y$ by lattice points, thus reducing the number of mid-lattice points of the partition. So we have a new partition satisfying (a) and (b) but has fewer endpoints being a mid-lattice point, a contradiction.
\end{proof}

Now we finally get a partition on which (3) holds and all its endpoints are lattice points. Thus, each part of the partition actually represent a word in $\Z^n$. We shall write $y_1,y_2,\ldots,y_m$ to be the words given by $\{ w_i \mid i\in S \}$ and $\{ w'_i \mid i\in T \}$, and $z_1,z_2,\ldots,z_m$ to be the words given by $\{ w_i \mid i\notin S \}$ and $\{ w'_i \mid i\notin T \}$. We let $y_j = \epsilon$ for $j > |T|+|S|$, and $z_k = \epsilon$ for $k > |\overline{T}|+|\overline{S}|$. Then, (3) says that $\sum_0^m \overline{y_i} = \sum_0^m \overline{z_i} = 0$ and condition (b) ensures that $|y_1\cdots y_m|, |z_1\cdots z_m| < |x_1\cdots x_m|$. Thus by induction hypothesis we have two valid sequences ending in $I(y_1,y_2,\ldots,y_m)$ and $I(z_1,z_2,\ldots,z_m)$, respectively. Since $y_i$'s and $z_i$'s is a refinement of $x_i$'s, there is also an instance of one of the last production rules that combines the $y_i$'s and $z_i$'s back to $x_i$'s. Concatenating the two sequences with this production rule will give a valid sequence ending in $I(x_1,x_2,\ldots,x_m)$. This completes the proof.
\end{proof}

We finish with the following corollary by combining our main Theorem with Theorem \ref{finite-index} to get a more general statement.

\begin{corollary}
The word problem of any finitely-generated virtually-abelian group is multiple context-free.
\end{corollary}

\bibliographystyle{amsalpha}
\bibliography{ref}

\end{document}